\documentclass{article}
\usepackage{graphicx}
\usepackage{fullpage}
\usepackage[usenames,dvipsnames]{color}
\usepackage{subfig}
\usepackage{amsmath,amsfonts,amsthm}
\usepackage{tikz}
\usepackage{dsfont}
\usetikzlibrary{spy}
\usepackage[section]{algorithm}
\usepackage{algorithmic}
\usepackage{ulem}
\usepackage{booktabs} 
\usepackage[title]{appendix}

\def\ver{2.3}
\newcommand{\update}[2]{\ifthenelse{\equal{#1}{\ver}}{\textcolor{red}{#2}}{#2}}
\theoremstyle{plain}
\newtheorem{theorem1}{Theorem}[section]
\newtheorem{proposition}[theorem1]{Proposition}

\theoremstyle{definition}

\newtheorem{example}[theorem1]{Example}

\newtheorem{remark}[theorem1]{Remark}

\numberwithin{table}{section}
\numberwithin{theorem1}{section}
\numberwithin{figure}{section}
\numberwithin{equation}{section}

\newcommand{\ZZ}{\mathrm{Z}}

\newcommand{\be}{\begin{equation}}
\newcommand{\ee}{\end{equation}}
\newcommand{\bea}{\begin{eqnarray}}
\newcommand{\eea}{\end{eqnarray}}
\newcommand{\bdm}{\begin{displaymath}}
\newcommand{\edm}{\end{displaymath}}
\newcommand{\beaq}{\begin{eqnarray*}}
\newcommand{\eeaq}{\end{eqnarray*}}

\newcommand{\R}{{\mathbb{R}}}

\newcommand{\II}{{\mathrm{I}}}
\newcommand{\EE}{{\mathrm{E}}}
\newcommand{\ii}{{\mathbf{i}}}
\newcommand{\jj}{{\mathbf{j}}}

\newcommand{\mm}{{\mathbf{m}}}
\newcommand{\NN}{{\mathbf{N}}}
\newcommand{\hh}{{\mathbf{h}}}
\newcommand{\fftd}{{\mathrm{FFT}_d}}

\newcommand{\ra}[1]{\renewcommand{\arraystretch}{#1}}
\providecommand{\keywords}[1]{\textbf{\textit{Keywords:}} #1}

\DeclareCaptionLabelFormat{andtable}{#1~#2  \&  \tablename~\thetable}
\usetikzlibrary{decorations.pathreplacing}

\graphicspath{{./Figures/}}
\begin{document}

\title{A Block Circulant Embedding Method for Simulation of Stationary Gaussian Random Fields on Block-regular Grids}
\author{M. Park\footnotemark[1] \ and M.V. Tretyakov\footnotemark[1]}
\date{\today}

\maketitle

\begin{abstract}
We propose a new method for sampling from stationary Gaussian random field on a grid which is not regular but has a regular block structure which is often the case in applications. The introduced block circulant embedding method (BCEM) can outperform the classical circulant embedding method (CEM) which requires a regularization of the irregular grid before its application. Comparison of BCEM vs CEM is performed on some typical model problems.
\end{abstract}

\renewcommand{\thefootnote}{\arabic{footnote}} \renewcommand{\thefootnote}{%
\fnsymbol{footnote}}
\footnotetext[1]{School of Mathematical Sciences, University of Nottingham, 
Nottingham, NG7 2RD, UK (emails: min.park@nottingham.ac.uk.and Michael.Tretyakov@nottingham.ac.uk.}

\section{Introduction}
Uncertainties are often modeled using stationary Gaussian fields \cite{delhomme:1979,gel:2004,marheineke:2006,skordos:2008,zhang:2010}. Efficient generation of samples from stationary Gaussian fields is crucial for using Monte Carlo techniques, which are the backbone of uncertainty quantification simulations, in studying behavior of systems subject to uncertainties. There are a few numerical techniques for sampling Gaussian random fields on a grid. For instance, one can find a square-root of the corresponding covariance matrix using Cholesky's decomposition and then multiply the square-root by a vector of independent Gaussian random variables to simulate a sample. This is an exact method but it is rarely used in applications due to the high cost of Cholesky's decomposition in high dimensions. Another possibility is the Karhunen-Loeve expansion  (see, e.g. \cite{ghanem:1991, maitre:2010}), which requires knowledge of eigenvalues and eigenfunctions of the covariance operator for the Gaussian random field. In many cases of practical interest the eigenvalue problem has to be solved numerically which can be expensive, especially when eigenvalues decay slowly. Also, this method is not exact. The fast and exact method of generating large samples from stationary Gaussian fields on regular grids is the circulant embedding method (CEM) \cite{wood:1994,dietrich:1997,wood:1999} which is widely used in various uncertainty quantification (UQ) applications such as groundwater flow simulation \cite{ghanem:1991,park:2014}, weather field forecasting \cite{gel:2004}, and liquid composite molding processes \cite{verleye:2012}. The two main drawbacks of CEM are (i) the requirement imposed on the grid to be regular while  irregular grids of a block structure naturally appear in many applications (see three typical examples below) and (ii) the need to deal with non-positive definiteness of circular embedding matrices which often occur in practical applications. The remedies for the latter were considered in \cite{stein:2002,gneiting:2006,helgason:2014}, here we deal with the first deficiency of CEM. To this end,
we propose a new block circulant embedding method (BCEM). Let us clarify the matter using the following three examples which come from sampling a random permeability field in groundwater flow simulations.

\begin{example}\label{ex:fem2d}
{\it Triangular finite element with a quadrature point located at the barycentre of the triangle.}
\end{example}
\vspace{-8pt}
\noindent Consider generation of a stationary log-normal random permeability field to be used in simulations based on triangular finite elements and the Gaussian quadrature rule of degree 1 within a rectangular domain. Assume that the rectangular domain consists of small rectangles (see Figure \ref{fig:2D_FEM}) and that there is no overlap of these rectangles. To perform the finite element simulation, it is sufficient to have sampled values of the permeability field at the quadrature points only (see the black circles in Fig.~\ref{fig:2D_FEM}). The covariance matrix of the corresponding stationary Gaussian random field at all quadrature (black) points is symmetric block-Toeplitz, but the blocks themselves are not symmetric. Hence, in order for the standard CEM to be applicable, 7 extra (i.e., artificial from the point of view of sampling permeability values sufficient for the finite element simulation) points should be added to each rectangle (the gray circles in Figure \ref{fig:2D_FEM} are regular grid points involving black circles). In contrast to CEM, the new method - BCEM - allows to sample values at the required points (black circles) without adding extra nodes to the grid and it does so in a very efficient way as we will see in the next sections.

\begin{figure}
\centering
\begin{tikzpicture}
\def\a{0}
\def\b{0}
\def\c{7.5}
\def\d{0}
\def\f{4.5}
\draw[-][very thick](\a,\b)--(\a+\f,\b);
\draw[-][very thick](\a,\b)--(\a,\b+\f);
\draw[-][very thick](\a,\b+\f)--(\a+\f,\b+\f);
\draw[-][very thick](\a+\f,\b)--(\a+\f,\b+\f);
\draw[-][very thick](\a,\b+\f)--(\a+\f,\b);
\node at (\a+\f*1/3,\b+\f*1/3)[circle,fill=black,inner sep=0pt,minimum size=4pt]{};
\node at (\a+\f*2/3,\b+\f*2/3)[circle,fill=black,inner sep=0pt,minimum size=4pt]{};
\draw[-][very thick](\c,\d)--(\c+\f,\d);
\draw[-][very thick](\c,\d)--(\c,\d+\f);
\draw[-][very thick](\c,\d+\f)--(\c+\f,\d+\f);
\draw[-][very thick](\c+\f,\d)--(\c+\f,\d+\f);
\draw[-][very thick](\c,\d+\f)--(\c+\f,\d);
\node at (\c+\f*0/3,\d+\f*0/3)[circle,fill=gray,inner sep=0pt,minimum size=8pt]{};
\node at (\c+\f*0/3,\d+\f*1/3)[circle,fill=gray,inner sep=0pt,minimum size=8pt]{};
\node at (\c+\f*0/3,\d+\f*2/3)[circle,fill=gray,inner sep=0pt,minimum size=8pt]{};
\node at (\c+\f*1/3,\d+\f*0/3)[circle,fill=gray,inner sep=0pt,minimum size=8pt]{};
\node at (\c+\f*1/3,\d+\f*1/3)[circle,fill=gray,inner sep=0pt,minimum size=8pt]{};
\node at (\c+\f*1/3,\d+\f*1/3)[circle,fill=black,inner sep=0pt,minimum size=4pt]{};
\node at (\c+\f*1/3,\d+\f*2/3)[circle,fill=gray,inner sep=0pt,minimum size=8pt]{};
\node at (\c+\f*2/3,\d+\f*0/3)[circle,fill=gray,inner sep=0pt,minimum size=8pt]{};
\node at (\c+\f*2/3,\d+\f*1/3)[circle,fill=gray,inner sep=0pt,minimum size=8pt]{};
\node at (\c+\f*2/3,\d+\f*2/3)[circle,fill=gray,inner sep=0pt,minimum size=8pt]{};
\node at (\c+\f*2/3,\d+\f*2/3)[circle,fill=black,inner sep=0pt,minimum size=4pt]{};
\end{tikzpicture}
\caption{Left: locations of nodes in 2D triangular elements. Right: regular grid nodes (gray) which contain all the original nodes (black). Note that the nodes on the right and top sides of the rectangles belong to the neighboring elements. } \label{fig:2D_FEM}
\end{figure}

\begin{table}
\centering
\ra{1.4}
\begin{tabular}{|c|c|c|c|}
\hline
dimension & BCEM & CEM & CEM/BCEM \\
\hline
1 & (2+1)$n$ & 4$n$ & 1.3 \\
\hline
2 & ($2^2+1$)$n$ & $4^2n$ & 3.2 \\
\hline
3 & $(2^3+1)n$ & $4^3n$& 7.1\\
\hline
\end{tabular}
\caption{Comparison of the number of nodes needed by BCEM and CEM in Example~\ref{ex:mlmc}.  } \label{tab:mlmc}
\end{table}

\begin{example}\label{ex:mlmc}
{\it Cell-centered finite volume discretization in multilevel Monte Carlo (MLMC) computation.}
\end{example}
\vspace{-8pt}
\noindent The multilevel Monte Carlo (MLMC) method is a Monte Carlo technique, which can give a substantial reduction of computational complexity in comparison with the standard Monte Carlo method thanks to making use of a hierarchical sampling \cite{cliffe:2011,barth:2011}. In the MLMC algorithm, when computing the difference of quantities on two consecutive grids with mesh sizes $h$ and $2h$, the pair of fine and coarse random samples must come from the same realization of the random field. In the cell-centered finite volume discretization, which uses permeability values at the centers of cells, the locations of coarse random filed do not coincide with nodes on the fine grid (see Fig.~\ref{fig:2D_MLMC}). In this case there exists a uniform grid with the mesh size $h/2$ containing both fine and coarse points, and hence it is possible to generate the required pair from the same realization by applying CEM on this finer uniform grid (grey circles in Fig.~\ref{fig:2D_MLMC}). However, this leads to an increase of both simulation time and memory requirements and, hence, to deterioration of the MLMC performance. Table \ref{tab:mlmc} compares the number of nodes at which the random field actually needed to be sampled for MLMC (which will be the same as the number of nodes used in BCEM) against the ones on the fine, regularized grid required by CEM. The portion of unused values grows as dimension increases. The benefit of BCEM is that exploiting the block-regular structure of grids used in MLMC, it allows us to sample at the points used in finite volume simulation without need to regularize the grid by adding extra points, which can result in substantial savings of both computational time and memory in comparison with applying CEM.

\begin{figure}
\centering
\begin{tikzpicture}
\def\a{0}
\def\b{0}
\def\c{7.5}
\def\d{0}
\def\f{4.5}
\draw[-][very thick](\a,\b)--(\a+\f,\b);
\draw[-][very thick](\a,\b)--(\a,\b+\f);
\draw[-][very thick](\a,\b+\f)--(\a+\f,\b+\f);
\draw[-][very thick](\a+\f,\b)--(\a+\f,\b+\f);
\draw[-][dashed](\a+\f*1/2,\b)--(\a+\f*1/2,\b+\f);
\draw[-][dashed](\a,\b+\f*1/2)--(\a+\f,\b+\f*1/2);
\draw [thick, black,decorate,decoration={brace,amplitude=10pt,mirror},xshift=0.4pt,yshift=-0.4pt](\a,\b) -- (\a+\f*1/2,\b) node[black,midway,yshift=-0.6cm] {\footnotesize $h$};
\draw [thick, black,decorate,decoration={brace,amplitude=10pt,mirror},xshift=0.4pt,yshift=-0.4pt](\a,\b-.2) -- (\a+\f,\b-.2) node[black,midway,yshift=-0.6cm] {\footnotesize $2h$};
\node at (\a+\f*1/4,\b+\f*1/4)[circle,fill=black,inner sep=0pt,minimum size=4pt]{};
\node at (\a+\f*3/4,\b+\f*1/4)[circle,fill=black,inner sep=0pt,minimum size=4pt]{};
\node at (\a+\f*1/4,\b+\f*3/4)[circle,fill=black,inner sep=0pt,minimum size=4pt]{};
\node at (\a+\f*3/4,\b+\f*3/4)[circle,fill=black,inner sep=0pt,minimum size=4pt]{};
\node at (\a+\f*2/4,\b+\f*2/4)[circle,draw,inner sep=0pt,minimum size=4pt]{};
\draw[-][very thick](\c,\d)--(\c+\f,\d);
\draw[-][very thick](\c,\d)--(\c,\d+\f);
\draw[-][very thick](\c,\d+\f)--(\c+\f,\d+\f);
\draw[-][very thick](\c+\f,\d)--(\c+\f,\d+\f);
\node at (\c+\f*0/4,\d+\f*0/4)[circle,fill=gray,inner sep=0pt,minimum size=8pt]{};
\node at (\c+\f*0/4,\d+\f*1/4)[circle,fill=gray,inner sep=0pt,minimum size=8pt]{};
\node at (\c+\f*0/4,\d+\f*2/4)[circle,fill=gray,inner sep=0pt,minimum size=8pt]{};
\node at (\c+\f*0/4,\d+\f*3/4)[circle,fill=gray,inner sep=0pt,minimum size=8pt]{};
\node at (\c+\f*1/4,\d+\f*0/4)[circle,fill=gray,inner sep=0pt,minimum size=8pt]{};
\node at (\c+\f*1/4,\d+\f*1/4)[circle,fill=gray,inner sep=0pt,minimum size=8pt]{};
\node at (\c+\f*1/4,\d+\f*1/4)[circle,fill=black,inner sep=0pt,minimum size=4pt]{};
\node at (\c+\f*1/4,\d+\f*2/4)[circle,fill=gray,inner sep=0pt,minimum size=8pt]{};
\node at (\c+\f*1/4,\d+\f*3/4)[circle,fill=gray,inner sep=0pt,minimum size=8pt]{};
\node at (\c+\f*1/4,\d+\f*3/4)[circle,fill=black,inner sep=0pt,minimum size=4pt]{};
\node at (\c+\f*2/4,\d+\f*0/4)[circle,fill=gray,inner sep=0pt,minimum size=8pt]{};
\node at (\c+\f*2/4,\d+\f*1/4)[circle,fill=gray,inner sep=0pt,minimum size=8pt]{};
\node at (\c+\f*2/4,\d+\f*2/4)[circle,fill=gray,inner sep=0pt,minimum size=8pt]{};
\node at (\c+\f*2/4,\d+\f*2/4)[circle,draw,inner sep=0pt,minimum size=4pt]{};
\node at (\c+\f*2/4,\d+\f*3/4)[circle,fill=gray,inner sep=0pt,minimum size=8pt]{};
\node at (\c+\f*3/4,\d+\f*0/4)[circle,fill=gray,inner sep=0pt,minimum size=8pt]{};
\node at (\c+\f*3/4,\d+\f*1/4)[circle,fill=gray,inner sep=0pt,minimum size=8pt]{};
\node at (\c+\f*3/4,\d+\f*1/4)[circle,fill=black,inner sep=0pt,minimum size=4pt]{};
\node at (\c+\f*3/4,\d+\f*2/4)[circle,fill=gray,inner sep=0pt,minimum size=8pt]{};
\node at (\c+\f*3/4,\d+\f*3/4)[circle,fill=gray,inner sep=0pt,minimum size=8pt]{};
\node at (\c+\f*3/4,\d+\f*3/4)[circle,fill=black,inner sep=0pt,minimum size=4pt]{};
\draw [thick, black,decorate,decoration={brace,amplitude=10pt,mirror},xshift=0.4pt,yshift=-0.4pt](\c,\d) -- (\c+\f*1/4,\d) node[black,midway,yshift=-0.6cm] {\footnotesize $h/2$};
\end{tikzpicture}
\caption{Left: location of sampling points on the fine grid (black) of size $h$ and the coarse grid (hollow) of size $2h$ using the cell-centered finite volume discretization in 2D. Right: uniform grid (gray) which contains both black and hollow points. Note that the nodes on the right and top sides of the rectangles belong to the neighboring elements. } \label{fig:2D_MLMC}
\end{figure}

\begin{example}\label{ex:cond}
{\it Conditional random field generation on block regular grids.}
\end{example}
\vspace{-8pt}
\noindent The conditional random field generation based on CEM was considered in \cite{dietrich:1996}. In this approach one builds a symmetric matrix of the form
\begin{equation} \label{eq:cond_circ}
R =
\left[
\begin{array}{c c}
R_{11} & R_{12} \\
R_{21} & R_{22}
\end{array}
\right],
\end{equation}
where $R_{11} \in$ \update{2.2}{$\mathbb{R}^{n_1\times n_1}$} is a (block) circulant matrix, and $R_{22}\in$ \update{2.2}{$\mathbb{R}^{n_2\times n_2}$} is a covariance matrix of filed values at locations of measurements, and generate random vectors using its square root
\begin{equation} \label{eq:cond_circ_half}
R^{1/2} =
\left[
\begin{array}{c c}
\frac{1}{\sqrt{n_1}}F\Lambda^{1/2} & 0 \\
K & L
\end{array}
\right],
\end{equation}
where  $K = \sqrt{n_1}R_{21}F\Lambda^{-1/2}$ and $L$ is a matrix such that $LL^T = R_{22}-KK^H$. Here $F$ denotes a discrete Fourier transform matrix, and $\Lambda$ is a diagonal matrix whose entries are eigenvalues of $R_{11}$. The computational costs of forming $\Lambda, K^H$ and $KK^H$ are $\mathcal{O}(n_1 \log n_1)$ flops, $\mathcal{O}(n_2n_1\log n_1)$ flops, and $\mathcal{O}(n_2^2n_1)$ flops, respectively.

BCEM can also be used for generation of random fields conditioned on observations. As with the unconditional sampling discussed above, applications of conditional sampling often deal with grids which are not regular but have a regular block structure
(see, e.g. conditional MLMC simulation in \cite{park:2014}). Since BCEM requires smaller $n_1$ value as it does in the unconditional case, BCEM in the conditional random field setting can outperform CEM.

\medskip

BCEM also has the remarkable feature that it is paralellizable in contrast to the standard CEM which is a serial algorithm (of course, CEM can exploit parallelism of the fast Fourier transform (FFT) but BCEM's main ingredient is also FFT and it can benefit from FFT parallelism as well), i.e., BCEM has a further significant advantage over CEM.

The rest of the paper is organized as follows. In Section~\ref{section:1d} we illustrate the idea of BCEM in the case of one-dimensional space. In Section~\ref{section:md} we present a multi-dimensional BCEM. Computational complexity of BCEM is discussed in Section~\ref{sec:comp}, where some numerical experiments comparing BCEM with the standard CEM show that already in 2D BCEM can be three time faster in sample generation than CEM.

\section{Illustration of the idea} \label{section:1d}

To illustrate the idea of BCEM, we start with presenting it in the 1D case.
\begin{figure}[H]
\centering
\begin{tikzpicture}
\draw[-][very thick](0,0)--(6,0);
\draw[very thick, dashed](6,0)--(8,0);
\draw[-][very thick](8,0)--(10,0);
\draw[thick] (0,-.1) node[below]{$x_0$} -- (0,0.1);
\draw[thick] (2,-.1) node[below]{$x_1$} -- (2,0.1);
\draw[thick] (4,-.1) node[below]{$x_2$} -- (4,0.1);
\draw[thick] (6,-.1) node[below]{$x_3$} -- (6,0.1);
\draw[thick] (8,-.1) node[below,xshift=.2cm]{$x_{N-1}$} -- (8,0.1);
\draw[thick] (10,-.1) node[below,xshift=.1cm]{$x_{N}$} -- (10,0.1);
\node at (0.33,0)[circle,fill=black,inner sep=0pt,minimum size=4pt]{};
\node at (0.33,0)[above]{$s^{(0)}_1$};
\node at (1.67,0)[circle,fill=black,inner sep=0pt,minimum size=4pt]{};
\node at (1.67,0)[above]{$s^{(0)}_2$};
\node at (2.33,0)[circle,fill=black,inner sep=0pt,minimum size=4pt]{};
\node at (2.33,0)[above]{$s^{(1)}_1$};
\node at (3.67,0)[circle,fill=black,inner sep=0pt,minimum size=4pt]{};
\node at (3.67,0)[above]{$s^{(1)}_2$};
\node at (4.33,0)[circle,fill=black,inner sep=0pt,minimum size=4pt]{};
\node at (4.33,0)[above]{$s^{(2)}_1$};
\node at (5.67,0)[circle,fill=black,inner sep=0pt,minimum size=4pt]{};
\node at (5.67,0)[above]{$s^{(2)}_2$};
\node at (8.33,0)[circle,fill=black,inner sep=0pt,minimum size=4pt]{};
\node at (8.33,0)[above, xshift=.3cm]{$s^{(N-1)}_1$};
\node at (9.67,0)[circle,fill=black,inner sep=0pt,minimum size=4pt]{};
\node at (9.67,0)[above, xshift=.3cm]{$s^{(N-1)}_2$};
\node at (1.0,0)[below]{$\tilde{s}^{(0)}_1$};
\node at (3.0,0)[below]{$\tilde{s}^{(1)}_1$};
\node at (5.0,0)[below]{$\tilde{s}^{(2)}_1$};
\node at (9.0,0)[below, xshift=.3cm]{$\tilde{s}^{(N-1)}_1$};
\node at (1,0)[circle,fill=gray,inner sep=0pt,minimum size=4pt]{};
\node at (3,0)[circle,fill=gray,inner sep=0pt,minimum size=4pt]{};
\node at (5,0)[circle,fill=gray,inner sep=0pt,minimum size=4pt]{};
\node at (9,0)[circle,fill=gray,inner sep=0pt,minimum size=4pt]{};
\end{tikzpicture}
\caption{1D uniform grid $\Omega_r=\{x_0,\ldots,x_N \} \in \Omega = [x_0, x_N]$.  Black circles represent the locations of points $s_{j}^{(i)} \in \Omega_s$, and
the combination of black and grey circles correspond to the uniform grid $ \tilde \Omega_s $. } \label{fig:1D_grid}
\end{figure}

Consider a uniform grid $\Omega_r=\{x_0,\ldots,x_N\}$ on the interval $\Omega = [x_0, x_N]$ with a grid size $h = (x_N-x_0)/N$, and sets of points $S_i = \{s_{1}^{(i)},\ldots,s_{\ell}^{(i)}\} \subset \Omega_i = [x_i, x_{i+1}]$ with $s_{j}^{(i)} = x_i + \delta_j$, where $0 \leq \delta_1 < \delta_2 < \cdots <\delta_\ell < h$ (see Figure \ref{fig:1D_grid}). Note that $\delta_j$ are independent of the index $i$. The grid $\Omega_s := \bigcup\limits_{i = 0}^{N-1}S_i$ is, in general, non-uniform (it is uniform if $\ell=1$) but it is block-uniform, i.e., the distribution of points in each sub-interval (in other words, block) $\Omega_i$ is the same.

Let $\ZZ(x)$, $x \in \R$, be a stationary Gaussian random field  with zero mean and covariance function $r(x)$. Our aim is to sample from $\ZZ(x)$ on the grid $\Omega_s$. If $\Omega_s$ is not a grid of equispaced points, then the covariance matrix of the field $\ZZ(x)$ on $\Omega_s$ is not Toeplitz. In this case the standard CEM \cite{wood:1994,dietrich:1997,wood:1999} cannot be applied to this covariance matrix in order to perform highly efficient computing of its square-root with subsequent generation of the required Gaussian field samples. The simplest remedy is to extend the non-uniform grid $\Omega_s$ to the uniform grid $\tilde \Omega_s$ by adding points (see Figure~\ref{fig:1D_grid}) and then apply the standard circulant embedding method, but this approach results in a substantial increase of computational costs. In this paper, we propose a different approach which does not need in adding points to $\Omega_s$ and which is cheaper than the use of the standard CEM on the extended uniform grid $\tilde \Omega_s$.

Consider the covariance matrix $R$ of the random vector $\ZZ(s_{j}^{(i)})$, $s_{j}^{(i)} \in \Omega_s$, written in the block matrix form:
\begin{equation} \label{eq:1d_covmat}
R =
\left[
\begin{array}{c c c c c}
R_{0,0} & R_{0,1} & R_{0,2} & \cdots & R_{0,N-1} \\
R_{1,0} & R_{1,1} & R_{1,2} &\cdots & R_{1,N-1}\\
R_{2,0} & R_{2,1} & R_{2,2} &\cdots & R_{2,N-1}\\
\vdots & \vdots & \vdots & \ddots & \vdots\\
R_{N-1,0} & R_{N-1,1} &R_{N-1,2} &\cdots& R_{N-1,N-1}
\end{array}
\right]_{N\ell \times N\ell},
\end{equation}
where each block matrix $R_{i,k}$ is defined as
\begin{equation}
R_{i,k} = \left[ r(|s^{(i)}_j-s^{(k)}_l|)\right]_{1 \leq j,l \leq \ell}.
\end{equation}

Now note that, by construction,
\begin{equation} \label{eq:property_block_toeplitz}
R_{i,j} =  \begin{cases}
 R_{k,l} & \text{if $j-i = l-k$,} \\
 R_{k,l}^T & \text{if $j-i = k-l$.}
\end{cases}
\end{equation}
Property (\ref{eq:property_block_toeplitz}) implies that the covariance matrix $R$ from (\ref{eq:1d_covmat}) can be uniquely determined by its first block row and hence it is symmetric and block Toeplitz, having identical blocks along diagonals. Then $R$ can be rewritten as
\begin{equation} \label{eq:1d_covmat2}
R =
\left[
\begin{array}{c c c c c}
R_{0,0} & R_{0,1} & R_{0,2} & \cdots & R_{0,N-1} \\
R_{0,1}^T & R_{0,0} & R_{0,1} &\cdots & R_{0,N-2}\\
R_{0,2}^T & R_{0,1}^T & R_{0,0} &\cdots & R_{0,N-3}\\
\vdots & \vdots & \vdots&\ddots & \vdots\\
R_{0,N-1}^T & R_{0,N-2}^T &R_{0,N-3}^T &\cdots& R_{0,0}
\end{array}
\right].
\end{equation}

We now illustrate how CEM \cite{wood:1994,dietrich:1997,wood:1999} can be extended so that its new version, BCEM, is applicable
to the symmetric block Toeplitz matrix $R$ from (\ref{eq:1d_covmat2}). To this end, we embed $R$ in the $m\ell \times m\ell$ symmetric block Toeplitz matrix $C$ for some even integer $m\geq 2N$:
\begin{equation} \label{eq:block_circ}
C = \left[
\begin{array}{c c c c}
C_0 & C_1 & \cdots & C_{m-1} \\
C_{m-1} & C_0 & \cdots & C_{m-2}\\
\vdots & \vdots & \ddots & \vdots\\
C_1 & C_2 & \cdots &C_0
\end{array}
\right] ,
\end{equation}where
\begin{equation}
C_k =  \left[ r(g(|s^{(0)}_i-s^{(k)}_j|))\right]_{1 \leq i,j \leq \ell}
\end{equation}
and
\begin{equation} \label{eq:def_ext_covf}
g(x) =  \begin{cases}
 x & \text{if $x < mh/2$,} \\
 mh - x & \text{if $x \geq mh/2$.}
\end{cases}
\end{equation}
Note that $C_i = R_{0,i}$ for $0\leq i \leq N-1$ and that $C$ is the covariance matrix for $\ZZ(x)$ defined in the circular manner on the grid $\Omega^E_s = \bigcup\limits_{i = 0}^{m-1}S_i \subset \Omega^E=[x_0,x_m]$, where $x_m=x_0+mh$ and $S_i$ are defined in the same way as before.

It is not difficult to see that the matrix $C$ has the following properties
\begin{align}
C_0 &= C_0^T, \label{eq:sym_c0}\\
C_k &= C_{m-k}^T,\text{ for } 1 \leq k \leq \frac{m}{2}. \label{eq:sym_cb}
\end{align}
The properties (\ref{eq:sym_c0}) and (\ref{eq:sym_cb}) imply that $C$ is a symmetric block circulant matrix.

Let $F_B$ be the tensor product of a one-dimensional discrete Fourier matrix $F_m^1$ of order $m$ and an identity matrix $\II_\ell$ of size $\ell\times \ell$:
\[
 F_B = F_m^1 \otimes \II_\ell =
 \left[
\begin{array}{c c c c c}
\II_\ell & \II_\ell & \II_\ell &\cdots & \II_\ell \\
\omega_0\II_\ell & \omega_1 \II_\ell& \omega_2 \II_\ell & \cdots & \omega_{m-1}\II_\ell\\
\omega_0^2\II_\ell & \omega_1^2 \II_\ell& \omega_2^2 \II_\ell & \cdots & \omega_{m-1}^2\II_\ell\\
\vdots&\vdots&\vdots&\ddots&\vdots\\
\omega_0^{m-1}\II_\ell & \omega_1^{m-1} \II_\ell& \omega_2^{m-1} \II_\ell & \cdots & \omega_{m-1}^{m-1}\II_\ell
\end{array}
\right].
\]
The matrix $C$ is unitarily block diagonalizable by $F_B$ \cite{davis:1979, wood:1999}, i.e., there exists $\ell \times \ell$ matrices $\Lambda_k$, $k = 0,1,\ldots, m-1$, such that
\begin{equation} \label{eq:block_factorisation}
 C  = \frac{1}{m} F_B\Lambda F_B^H, \quad \text{where } \Lambda = \left[
\begin{array}{c c c c}
\Lambda_0 & 0 & \cdots & 0 \\
0 & \Lambda_1 & \cdots & 0\\
\vdots&\vdots&\ddots&\vdots\\
0 & 0 & \cdots & \Lambda_{m-1}
\end{array}
\right].
\end{equation}
Here $H$ denotes the conjugate transpose. Similarly to the eigenvalue decomposition of a symmetric circulant matrix whose eigenvalues can be calculated by performing a discrete Fourier transform of its first row (or column), the block matrices on the diagonal of $\Lambda$ can be computed as
\begin{equation}\label{eq:eigdecmp_bc}
\left[C_0 \quad C_1 \cdots C_{m-1} \right] F_B = \left[\Lambda_0 \quad \Lambda_1 \cdots \Lambda_{m-1} \right]
\end{equation}
or in the component-wise form:
\begin{equation} \label{eq:eigdecmp_pointwise}
\left[C_0^{i,j} \quad C_1^{i,j} \cdots C_{m-1}^{i,j} \right] F_m^1 = \left[\Lambda_0^{i,j} \quad \Lambda_1^{i,j} \cdots \Lambda_{m-1}^{i,j} \right], \text{ where } 1 \leq i,j \leq \ell.
\end{equation}

Since the block circulant matrix $C$ is real and symmetric, $\Lambda_k$ are Hermitian. Furthermore, all the diagonal elements of $\Lambda_k$ are equal. Therefore, only  $\ell(\ell+1)/2 - (\ell-1)$ applications of $F_m^1$ are required for computing $\Lambda$.

\begin{remark}
Consider a uniform grid, i.e., a block-regular grid with the size of regular grid being a multiple of the number of blocks or, in other words, the points in each
$S_i$ being uniformly located. Then BCEM is applicable on the uniform grid (recall that CEM works on regular grids only). Since the covariance depends on the distance between points only, the block circulant matrix $C$ on the uniform grid satisfies the relationship
\[
C^{a,b}_k = C^{c,d}_k \text{ if }|a-b| = |c-d|.
\]
Consequently, $\Lambda_k$ in (\ref{eq:eigdecmp_bc}) are Toeplitz and the number of 1D FFT $F_m^1$ to compute distinctive values of $\Lambda$ is equal to $\ell$.
Thus, in BCEM the block circulant matrix can be diagonalized by using $\ell$ FFTs of order $m$ followed by using a Cholesky decomposition of the block diagonal matrix $\Lambda$, whose block entries are of size $\ell\times\ell$. For small $\ell$, the overall computational cost is dominated by $\mathcal{O}(\ell m \log_2 m)$. On the other hand, the complexity of CEM is dominated by FFTs of order $m\ell$, which gives the overall cost $\mathcal{O}(m\ell \update{2.2}{\log_2} m\ell)$. Hence, BCEM can outperform CEM on the uniform grid, where both CEM and BCEM use the same covariance matrix (see also Remark~\ref{rem41}).
\end{remark}

The symmetricity of $C$ also guaranties the spectral decomposition
\begin{equation} \label{eq:decomp_lambdak}
\Lambda_k = U_k D_k U_k^H,
\end{equation}
where $U_k$ is unitary and $D_k$ is a real-valued diagonal matrix. The following proposition implies that $\Lambda$ from (\ref{eq:block_factorisation}) can be decomposed with $m/2+1$ applications of the spectral decompositions (\ref{eq:decomp_lambdak}).

\begin{proposition} \label{prop:symmetricity}
The block diagonal matrix
$\Lambda$ from (\ref{eq:block_factorisation}) has the property
\begin{equation}\label{eq:symmetricity}
\Lambda_k = \overline{\Lambda_{m-k}} \text{, for } 1 \leq k \leq \frac{m}{2},
\end{equation}
where the bar denotes the matrix with conjugate complex entries.
\end{proposition}
\begin{proof}
Let $\omega_n = \exp(\frac{2\pi n}{m}i) $ be a root of unity. Then (see (\ref{eq:sym_cb}) and (\ref{eq:eigdecmp_pointwise})):
\begin{align*}
\overline{\Lambda_{m-k}} &= C_0 + \overline{\omega_{m-k}}C_1 + \cdots + \overline{\omega_{m-k}}^{m-1}C_{m-1} \\
&= C_0 + \omega_{k}C_1 + \cdots + \omega_{k}^{m-1}C_{m-1} \\
& = \Lambda_k.
\end{align*}
\end{proof}

It follows from (\ref{eq:block_factorisation}) and (\ref{eq:decomp_lambdak}) that $C$ has the eigenvalue decomposition
\begin{equation}\label{eq:eigdecomp}
C = 	\frac{1}{m}\left(F_BU\right)D\left(F_BU\right)^H,
\end{equation}
where the unitary block-diagonal matrix $U$ and the diagonal matrix $D$ are of the form
\[
 U = \left[
\begin{array}{c c c c}
U_0 & 0 & \cdots & 0 \\
0 & U_1 & \cdots & 0\\
\vdots&\vdots&\ddots&\vdots\\
0 & 0 & \cdots & U_{m-1}
\end{array}
\right]
\text{ and }
D = \left[
\begin{array}{c c c c}
D_0 & 0 & \cdots & 0 \\
0 & D_1 & \cdots & 0\\
\vdots&\vdots&\ddots&\vdots\\
0 & 0 & \cdots & D_{m-1}
\end{array}
\right].
\]
We note that C is non-negative definite if and only if $D_k^{i,i} \geq 0$ for each $0\leq k \leq m-1$ and $1 \leq i \leq \ell$.

Assume for the moment that all the eigenvalues of $C$ are non-negative. Let two independent random vectors $\xi_1$ and $\xi_2$, each of size $m$, be normally distributed $\mathcal{N}(O,\II_m)$, i.e., $\EE[\xi_i \xi_j^T] = \delta_{ij}\II_m$, where $\delta_{ij}$ denotes the Kronecker delta. Set $\eta =U(D/m)^{\frac{1}{2}}(\xi_1+i\xi_2)$. Then the real and imaginary parts of the vector $\zeta:=F_B\eta$ give two independent random vectors $\zeta_1$ and $\zeta_2$ that are both distributed as $N(0,C)$. Since $R$ is embedded in $C$, the corresponding parts of $\zeta_1$ and $\zeta_2$ are distributed as $\mathcal{N}(O,R)$. Note that the matrix-vector multiplication $F_B \eta$ can be calculated component-wise by $\ell$ applications of $F_m^1$.

The algorithm described above depends on non-negative definiteness of the symmetric block circulant matrix $C$. The sufficient conditions for symmetric circulant matrices to have all non-negative eigenvalues were developed for 1D case in \cite{dietrich:1997} and \cite{wood:1994}. Here we extend these conditions to the symmetric block circulant matrix $C$ from (\ref{eq:block_circ}). To this end, introduce a uniform grid $\tilde\Omega_s$ such that $\Omega_s \subset \tilde\Omega_s$ and consider the covariance matrix $\tilde R$ defined on $\tilde\Omega_s$.  As $\Omega_s$ is a subset of $\tilde\Omega_s$, $R$ is a sub-matrix of the matrix $\tilde R$. Let a uniform grid $\tilde\Omega^E_s$ contain all points of $\Omega_s^E = \bigcup\limits_{i = 0}^{m-1}S_i$. Then the symmetric block circulant matrix $C$ is a sub-matrix of a symmetric circulant matrix $\tilde C$:
\begin{equation}
\tilde C^{i,j} = \left[r(g(|x_i-x_j|))\right],
\end{equation}
where the function $g(x)$ is as in (\ref{eq:def_ext_covf}) and $x_i, x_j \in \tilde \Omega^E_s$. Therefore, there exists an injection matrix $P^T$ such that
\begin{equation} \label{eq:relationC}
C = P^T\tilde{C}P.
\end{equation}
 An injection matrix can be built by eliminating rows of the identity matrix, which correspond to points not in $\Omega^E_s$. For instance,
\[
  \left[
\begin{array}{c c c c c}
1 & 0 & 0 & 0 & 0 \\
0 & 0 & 1 & 0 & 0 \\
0 & 0 & 0 & 0 & 1
\end{array}
\right]
\]
is an injection matrix from $\{x_1, x_2, x_3, x_4, x_5 \}$  to $\{x_1, x_3, x_5\}$.
The relationship (\ref{eq:relationC}) leads to the following proposition.

\begin{proposition}
If $\tilde C$ is non-negative definite, then so is $C$.
\end{proposition}

When the circulant matrix $\tilde C$ fails to be non-negative definite, Wood and Chen \cite{wood:1994} suggested to increase the size of $\tilde C$ until it becomes non-negative definite (the so-called padding technique). From the relationship (\ref{eq:relationC}) between $C$ and $\tilde C$, the same strategy can be used for the matrix $C$. That is, increase $m$ until $C$ becomes non-negative definite. Therefore, the number of blocks $m$, which is required for $C$ to be non-negative definite, depends on the grid size of the uniform grid $\tilde \Omega^E_s$, not on the number of points in $\Omega_s$. Thus, the number of paddings needed for BCEM is the same as for CEM (see also Remark~\ref{rempad}).

\section{Multidimensional BCEM}\label{section:md}
In the previous section we illustrated the idea of BCEM in the simpler setting of 1D space.
In this section we present multi-dimensional BCEM which computational complexity is discussed in the next section.

\begin{figure}[h!]
\centering
\begin{tikzpicture}
\tikzstyle{spt} = [circle,fill=black,inner sep=0pt,minimum size=2pt]

\filldraw[fill=gray!10!white,fill opacity=0.1] (0,0) rectangle (11.99,11.99);
\filldraw[fill=gray!10] (0,0) rectangle (5.99,5.99);
\draw[step=1.5cm,gray,very thin] (0,0) grid (12,12);

\node at (0,0)[]{$x_{\jj_0}$};
\node at (1.5,0)[]{$x_{\jj_1}$};
\node at (3,0)[]{$x_{\jj_2}$};
\node at (4.5,0)[]{$x_{\jj_3}$};
\node at (6,0)[]{$x_{\jj_4}$};
\node at (7.5,0)[]{$x_{\jj_5}$};
\node at (9,0)[]{$x_{\jj_6}$};
\node at (10.5,0)[]{$x_{\jj_7}$};
\node at (12,0)[]{$x_{\jj_8}$};

\node at (0,1.5)[]{$x_{\jj_9}$};
\node at (1.5,1.5)[]{$x_{\jj_{10}}$};
\node at (3,1.5)[]{$x_{\jj_{11}}$};
\node at (4.5,1.5)[]{$x_{\jj_{12}}$};
\node at (6,1.5)[]{$x_{\jj_{13}}$};
\node at (7.5,1.5)[]{$x_{\jj_{14}}$};
\node at (9,1.5)[]{$x_{\jj_{15}}$};
\node at (10.5,1.5)[]{$x_{\jj_{16}}$};
\node at (12,1.5)[]{$x_{\jj_{17}}$};

\node at (0,3)[]{$x_{\jj_{18}}$};
\node at (1.5,3)[]{$x_{\jj_{19}}$};
\node at (3,3)[]{$x_{\jj_{20}}$};
\node at (4.5,3)[]{$x_{\jj_{21}}$};
\node at (6,3)[]{$x_{\jj_{22}}$};
\node at (7.5,3)[]{$x_{\jj_{23}}$};
\node at (9,3)[]{$x_{\jj_{24}}$};
\node at (10.5,3)[]{$x_{\jj_{25}}$};
\node at (12,3)[]{$x_{\jj_{26}}$};

\node at (0,4.5)[]{$x_{\jj_{27}}$};
\node at (1.5,4.5)[]{$x_{\jj_{28}}$};
\node at (3,4.5)[]{$x_{\jj_{29}}$};
\node at (4.5,4.5)[]{$x_{\jj_{30}}$};
\node at (6,4.5)[]{$x_{\jj_{31}}$};
\node at (7.5,4.5)[]{$x_{\jj_{32}}$};
\node at (9,4.5)[]{$x_{\jj_{33}}$};
\node at (10.5,4.5)[]{$x_{\jj_{34}}$};
\node at (12,4.5)[]{$x_{\jj_{35}}$};

\node at (0,6)[]{$x_{\jj_{36}}$};
\node at (1.5,6)[]{$x_{\jj_{37}}$};
\node at (3,6)[]{$x_{\jj_{38}}$};
\node at (4.5,6)[]{$x_{\jj_{39}}$};
\node at (6,6)[]{$x_{\jj_{40}}$};
\node at (7.5,6)[]{$x_{\jj_{41}}$};
\node at (9,6)[]{$x_{\jj_{42}}$};
\node at (10.5,6)[]{$x_{\jj_{43}}$};
\node at (12,6)[]{$x_{\jj_{44}}$};

\node at (0,7.5)[]{$x_{\jj_{45}}$};
\node at (1.5,7.5)[]{$x_{\jj_{46}}$};
\node at (3,7.5)[]{$x_{\jj_{47}}$};
\node at (4.5,7.5)[]{$x_{\jj_{48}}$};
\node at (6,7.5)[]{$x_{\jj_{49}}$};
\node at (7.5,7.5)[]{$x_{\jj_{50}}$};
\node at (9,7.5)[]{$x_{\jj_{51}}$};
\node at (10.5,7.5)[]{$x_{\jj_{52}}$};
\node at (12,7.5)[]{$x_{\jj_{53}}$};

\node at (0,9)[]{$x_{\jj_{54}}$};
\node at (1.5,9)[]{$x_{\jj_{55}}$};
\node at (3,9)[]{$x_{\jj_{56}}$};
\node at (4.5,9)[]{$x_{\jj_{57}}$};
\node at (6,9)[]{$x_{\jj_{58}}$};
\node at (7.5,9)[]{$x_{\jj_{59}}$};
\node at (9,9)[]{$x_{\jj_{60}}$};
\node at (10.5,9)[]{$x_{\jj_{61}}$};
\node at (12,9)[]{$x_{\jj_{62}}$};

\node at (0,10.5)[]{$x_{\jj_{63}}$};
\node at (1.5,10.5)[]{$x_{\jj_{64}}$};
\node at (3,10.5)[]{$x_{\jj_{65}}$};
\node at (4.5,10.5)[]{$x_{\jj_{66}}$};
\node at (6,10.5)[]{$x_{\jj_{67}}$};
\node at (7.5,10.5)[]{$x_{\jj_{68}}$};
\node at (9,10.5)[]{$x_{\jj_{69}}$};
\node at (10.5,10.5)[]{$x_{\jj_{70}}$};
\node at (12,10.5)[]{$x_{\jj_{71}}$};

\node at (0,12)[]{$x_{\jj_{72}}$};
\node at (1.5,12)[]{$x_{\jj_{73}}$};
\node at (3,12)[]{$x_{\jj_{74}}$};
\node at (4.5,12)[]{$x_{\jj_{75}}$};
\node at (6,12)[]{$x_{\jj_{76}}$};
\node at (7.5,12)[]{$x_{\jj_{77}}$};
\node at (9,12)[]{$x_{\jj_{78}}$};
\node at (10.5,12)[]{$x_{\jj_{79}}$};
\node at (12,12)[]{$x_{\jj_{80}}$};

\node at (0,0)[above,xshift=.3cm,yshift=-.1cm]{$\Omega_{\jj_0}$};
\node at (1.5,0)[above,xshift=.3cm,yshift=-.1cm]{$\Omega_{\jj_1}$};
\node at (3,0)[above,xshift=.3cm,yshift=-.1cm]{$\Omega_{\jj_2}$};
\node at (4.5,0)[above,xshift=.3cm,yshift=-.1cm]{$\Omega_{\jj_3}$};
\node at (6,0)[above,xshift=.3cm,yshift=-.1cm]{$\Omega_{\jj_4}$};
\node at (7.5,0)[above,xshift=.3cm,yshift=-.1cm]{$\Omega_{\jj_5}$};
\node at (9,0)[above,xshift=.3cm,yshift=-.1cm]{$\Omega_{\jj_6}$};
\node at (10.5,0)[above,xshift=.3cm,yshift=-.1cm]{$\Omega_{\jj_7}$};

\node at (0,1.5)[above,xshift=.3cm,yshift=-.1cm]{$\Omega_{\jj_{9}}$};
\node at (1.5,1.5)[above,xshift=.3cm,yshift=-.1cm]{$\Omega_{\jj_{10}}$};
\node at (3,1.5)[above,xshift=.3cm,yshift=-.1cm]{$\Omega_{\jj_{11}}$};
\node at (4.5,1.5)[above,xshift=.3cm,yshift=-.1cm]{$\Omega_{\jj_{12}}$};
\node at (6,1.5)[above,xshift=.3cm,yshift=-.1cm]{$\Omega_{\jj_{13}}$};
\node at (7.5,1.5)[above,xshift=.3cm,yshift=-.1cm]{$\Omega_{\jj_{14}}$};
\node at (9,1.5)[above,xshift=.3cm,yshift=-.1cm]{$\Omega_{\jj_{15}}$};
\node at (10.5,1.5)[above,xshift=.3cm,yshift=-.1cm]{$\Omega_{\jj_{16}}$};

\node at (0,3)[above,xshift=.3cm,yshift=-.1cm]{$\Omega_{\jj_{18}}$};
\node at (1.5,3)[above,xshift=.3cm,yshift=-.1cm]{$\Omega_{\jj_{19}}$};
\node at (3,3)[above,xshift=.3cm,yshift=-.1cm]{$\Omega_{\jj_{20}}$};
\node at (4.5,3)[above,xshift=.3cm,yshift=-.1cm]{$\Omega_{\jj_{21}}$};
\node at (6,3)[above,xshift=.3cm,yshift=-.1cm]{$\Omega_{\jj_{22}}$};
\node at (7.5,3)[above,xshift=.3cm,yshift=-.1cm]{$\Omega_{\jj_{23}}$};
\node at (9,3)[above,xshift=.3cm,yshift=-.1cm]{$\Omega_{\jj_{24}}$};
\node at (10.5,3)[above,xshift=.3cm,yshift=-.1cm]{$\Omega_{\jj_{25}}$};

\node at (0,4.5)[above,xshift=.3cm,yshift=-.1cm]{$\Omega_{\jj_{27}}$};
\node at (1.5,4.5)[above,xshift=.3cm,yshift=-.1cm]{$\Omega_{\jj_{28}}$};
\node at (3,4.5)[above,xshift=.3cm,yshift=-.1cm]{$\Omega_{\jj_{29}}$};
\node at (4.5,4.5)[above,xshift=.3cm,yshift=-.1cm]{$\Omega_{\jj_{30}}$};
\node at (6,4.5)[above,xshift=.3cm,yshift=-.1cm]{$\Omega_{\jj_{31}}$};
\node at (7.5,4.5)[above,xshift=.3cm,yshift=-.1cm]{$\Omega_{\jj_{32}}$};
\node at (9,4.5)[above,xshift=.3cm,yshift=-.1cm]{$\Omega_{\jj_{33}}$};
\node at (10.5,4.5)[above,xshift=.3cm,yshift=-.1cm]{$\Omega_{\jj_{34}}$};

\node at (0,6)[above,xshift=.3cm,yshift=-.1cm]{$\Omega_{\jj_{36}}$};
\node at (1.5,6)[above,xshift=.3cm,yshift=-.1cm]{$\Omega_{\jj_{37}}$};
\node at (3,6)[above,xshift=.3cm,yshift=-.1cm]{$\Omega_{\jj_{38}}$};
\node at (4.5,6)[above,xshift=.3cm,yshift=-.1cm]{$\Omega_{\jj_{39}}$};
\node at (6,6)[above,xshift=.3cm,yshift=-.1cm]{$\Omega_{\jj_{40}}$};
\node at (7.5,6)[above,xshift=.3cm,yshift=-.1cm]{$\Omega_{\jj_{41}}$};
\node at (9,6)[above,xshift=.3cm,yshift=-.1cm]{$\Omega_{\jj_{42}}$};
\node at (10.5,6)[above,xshift=.3cm,yshift=-.1cm]{$\Omega_{\jj_{43}}$};

\node at (0,7.5)[above,xshift=.3cm,yshift=-.1cm]{$\Omega_{\jj_{45}}$};
\node at (1.5,7.5)[above,xshift=.3cm,yshift=-.1cm]{$\Omega_{\jj_{46}}$};
\node at (3,7.5)[above,xshift=.3cm,yshift=-.1cm]{$\Omega_{\jj_{47}}$};
\node at (4.5,7.5)[above,xshift=.3cm,yshift=-.1cm]{$\Omega_{\jj_{48}}$};
\node at (6,7.5)[above,xshift=.3cm,yshift=-.1cm]{$\Omega_{\jj_{49}}$};
\node at (7.5,7.5)[above,xshift=.3cm,yshift=-.1cm]{$\Omega_{\jj_{50}}$};
\node at (9,7.5)[above,xshift=.3cm,yshift=-.1cm]{$\Omega_{\jj_{51}}$};
\node at (10.5,7.5)[above,xshift=.3cm,yshift=-.1cm]{$\Omega_{\jj_{52}}$};

\node at (0,9)[above,xshift=.3cm,yshift=-.1cm]{$\Omega_{\jj_{54}}$};
\node at (1.5,9)[above,xshift=.3cm,yshift=-.1cm]{$\Omega_{\jj_{55}}$};
\node at (3,9)[above,xshift=.3cm,yshift=-.1cm]{$\Omega_{\jj_{56}}$};
\node at (4.5,9)[above,xshift=.3cm,yshift=-.1cm]{$\Omega_{\jj_{57}}$};
\node at (6,9)[above,xshift=.3cm,yshift=-.1cm]{$\Omega_{\jj_{58}}$};
\node at (7.5,9)[above,xshift=.3cm,yshift=-.1cm]{$\Omega_{\jj_{59}}$};
\node at (9,9)[above,xshift=.3cm,yshift=-.1cm]{$\Omega_{\jj_{60}}$};
\node at (10.5,9)[above,xshift=.3cm,yshift=-.1cm]{$\Omega_{\jj_{61}}$};

\node at (0,10.5)[above,xshift=.3cm,yshift=-.1cm]{$\Omega_{\jj_{63}}$};
\node at (1.5,10.5)[above,xshift=.3cm,yshift=-.1cm]{$\Omega_{\jj_{64}}$};
\node at (3,10.5)[above,xshift=.3cm,yshift=-.1cm]{$\Omega_{\jj_{65}}$};
\node at (4.5,10.5)[above,xshift=.3cm,yshift=-.1cm]{$\Omega_{\jj_{66}}$};
\node at (6,10.5)[above,xshift=.3cm,yshift=-.1cm]{$\Omega_{\jj_{67}}$};
\node at (7.5,10.5)[above,xshift=.3cm,yshift=-.1cm]{$\Omega_{\jj_{68}}$};
\node at (9,10.5)[above,xshift=.3cm,yshift=-.1cm]{$\Omega_{\jj_{69}}$};
\node at (10.5,10.5)[above,xshift=.3cm,yshift=-.1cm]{$\Omega_{\jj_{70}}$};

\foreach \x in {0,1.5,...,10.5}{
	\foreach \y in {0,1.5,...,10.5}{
		\node at (\x+.5,\y+.5)[spt]{};
		\node at (\x+1,\y+.5)[spt]{};
		\node at (\x+.5,\y+1)[spt]{};
		\node at (\x+1,\y+1)[spt]{};
	}
}

\def\a{3}
\def\b{-3.5}
\filldraw[fill=gray!10] (\a,\b) rectangle (\a+6,\b+3);
\draw[-][thick](\a,\b)--(\a+6,\b);
\draw[-][thick](\a+6,\b+3)--(\a,\b+3);
\draw[-][thick](\a,\b)--(\a,\b+3);
\draw[-][thick](\a+6,\b+3)--(\a+6,\b);
\draw[-][thick](\a+3,\b+3)--(\a+3,\b);
\node at (\a+1,\b+1)[spt]{};
\node at (\a+1,\b+2)[spt]{};
\node at (\a+2,\b+1)[spt]{};
\node at (\a+2,\b+2)[spt]{};
\node at (\a+4,\b+1)[spt]{};
\node at (\a+4,\b+2)[spt]{};
\node at (\a+5,\b+1)[spt]{};
\node at (\a+5,\b+2)[spt]{};
\node at (\a+1,\b+1)[xshift=-.2cm,yshift=.1cm]{$s^{\jj_0}_1$};
\node at (\a+1,\b+2)[xshift=-.2cm,yshift=.1cm]{$s^{\jj_0}_3$};
\node at (\a+2,\b+1)[xshift=-.2cm,yshift=.1cm]{$s^{\jj_0}_2$};
\node at (\a+2,\b+2)[xshift=-.2cm,yshift=.1cm]{$s^{\jj_0}_4$};
\node at (\a+4,\b+1)[xshift=-.2cm,yshift=.1cm]{$s^{\jj_1}_1$};
\node at (\a+4,\b+2)[xshift=-.2cm,yshift=.1cm]{$s^{\jj_1}_3$};
\node at (\a+5,\b+1)[xshift=-.2cm,yshift=.1cm]{$s^{\jj_1}_2$};
\node at (\a+5,\b+2)[xshift=-.2cm,yshift=.1cm]{$s^{\jj_1}_4$};
\draw (1.5,0.75) ellipse (2cm and 1cm) ;
\node at (\a,\b)[above,xshift=.5cm,yshift=-.1cm]{$\Omega_{\jj_{0}}$};
\node at (\a+3,\b)[above,xshift=.5cm,yshift=-.1cm]{$\Omega_{\jj_{1}}$};
\node at (\a,\b)[]{$x_{\jj_{0}}$};
\node at (\a+3,\b)[]{$x_{\jj_{1}}$};
\node at (\a+6,\b)[]{$x_{\jj_{2}}$};
\node at (\a,\b+3)[]{$x_{\jj_{9}}$};
\node at (\a+3,\b+3)[]{$x_{\jj_{10}}$};
\node at (\a+6,\b+3)[]{$x_{\jj_{11}}$};
\node at (1.5,-.25)(start){};
\node at (\a,\b+1.5)(end){};
\draw[->](start).. controls ([xshift=2cm] start) and ([xshift=-1cm] end) ..(end);
\end{tikzpicture}
\caption{The 2D uniform grid  $\Omega_r^E=\{x_{\jj_0}{=}x_\textbf{0}, x_{\jj_1},\ldots,x_{\jj_{\overline {\mm+\mathbf{1}}-1}}{=}x_\mm\}$ on the rectangle $\Omega^E = [x_\textbf{0}[1], x_\mm[1]]\times [x_\textbf{0}[2], x_\mm[2]]$ for $\mm = (8,8)^T$. Nodes $s^{(\jj_k)}_j$ of the block-regular grid $\Omega_s^E$ are represented by black circles. The shaded rectangle corresponds to the computation (i.e, the domain of interest for the problem at hands) domain $\Omega = [x_\textbf{0}[1], x_\NN[1]]\times [x_\textbf{0}[2], x_\NN[2]]$ for $\NN = (4,4)^T$.} \label{fig:2D_grid}
\end{figure}

We start with introducing the notation which largely follows \cite{wood:1999}. Let $\mathbb{Z}^d$ be the set of $d$-vectors with non-negative integer components and $\mathbf{0}$ and $\mathbf{1}$ be the $d$-dimensional vectors whose all components equal to 0 and 1, respectively. For $\ii=(\ii[1],\ldots,\ii[d])^T$, $\jj=(\jj[1],\ldots,\jj[d])^T  \in  \mathbb{Z}^d$, we define addition in $\mathbb{Z}^d$:
\[
\ii + \jj = (\ii[1]+\jj[1], \ldots, \ii[d]+\jj[d])^T
\]
and also the product of elements of $\ii$:
\[
\overline{\ii} = \prod\limits_{k = 1}^d \ii[k].
\]
For any $\jj \in \mathbb{Z}^d$ all components of which are strictly positive, we define the set $\mathcal{I}(\jj)$:
\begin{equation}
\mathcal{I}(\jj) = \{\ii \in \mathbb{Z}^d:0 \leq \ii[k] \leq \jj[k] - 1 \mbox{ for all } 1 \leq k \leq d\}.
\end{equation}
Note that the cardinality of $\mathcal{I}(\jj)$ is equal to $\overline{\jj}$.

Introduce a $d$-dimensional rectangular parallelepiped
\begin{equation}\label{eq:omega1}
\Omega = \prod\limits_{i = 1}^d [x_\textbf{0}[i], x_{\NN}[i]] \subset \R^d,
\end{equation}
where $\NN = (\NN[1],\ldots,\NN[d])^T \in \mathbb{Z}^d$ and the vector $\hh$ with the components $\hh[i] =  (x_{\NN}[i] - x_\textbf{0}[i])/\NN[i]$.  Further, points $x_{\ii_k}=(x_{\ii_k}[1],\ldots,x_{\ii_k}[d])^T$ with $x_{\ii_k}[j] =x_\textbf{0}[j]+\ii_k[j]\hh[j]$ form a
regular grid   $\Omega_r=\{x_{\ii_0},\ldots,x_{\ii_{\overline {\NN+\mathbf{1}}-1}}\}$ on the rectangular parallelepiped $\Omega$ (see Fig.~\ref{fig:2D_grid}).
The domain $\Omega$ in (\ref{eq:omega1}) can be divided into $d$-dimensional rectangular parallelepipeds as
\begin{equation}
\Omega = \bigcup\limits_{\jj_k\in \mathcal{I}(\NN)} \Omega_{\jj_k},
\end{equation}
where
\begin{equation}
\Omega_{\jj_k} = \prod\limits_{i = 1}^d [x_{\jj_k}[i],  x_{\jj_k}[i]+\hh[i]], \, \jj_k \in \mathcal{I}(\NN).
\label{eq:omjj}
\end{equation}
For the purpose of algorithm development, we use a lexicographic ordering of $\jj_k$ with respect to $k$, i.e., row after row and layer after layer (see Fig.~\ref{fig:2D_grid}).

Consider a stationary Gaussian random field $\ZZ(x)$, $x \in \R^d$,  with zero mean and covariance function $r(x)$.
We assume that the problem at our hands is such that we need to sample $\ZZ(x)$ at the nodes $s^{(\jj_k)}_i$ defined as follows (see the examples in the Introduction and also Fig.~\ref{fig:2D_grid}):
\begin{equation}
s^{(\jj_k)}_j = x_{\jj_k} + \delta_j,
\end{equation}
where $\delta_j = (\delta_j[1],\ldots,\delta_j[d])^T$ with $0 \leq \delta_j[i] < \hh[i]$ for $1 \leq i \leq d$. Here $\ell$ is the number of sampling points in each subdomain $\Omega_{\jj_k}$. That is, in each $\Omega_{\jj_k}$ the points from the set $S_{\jj_k} = \{s^{(\jj_k)}_1,\ldots,s^{(\jj_k)}_\ell\}$ are distributed according to the same pattern for all $\jj_k \in \mathcal{I}(\NN)$. Denote the grid:
$$\Omega_s=\bigcup\limits_{\jj_k \in \mathcal{I}(\NN)}\Omega_{\jj_k}.$$

Note that $\delta_j$ are independent of the index vector $\jj_k$. The covariance matrix $R$ of $Z(s^{(\jj_k)}_i)$,  $s^{(\jj_k)}_i \in \Omega_s$, is block-Toeplitz.
In the one-dimensional-case it consists of non-Toeplitz blocks of order $\ell$ (see Section~\ref{section:1d}).
In the $d$-dimensional case with $d>1$, it consists of blocks which have all the properties of a correlation matrix in the $d-1$ dimensional space.
We emphasize that the matrix $R$ is not Toeplitz and hence CEM is not directly applicable here.

Analogously to CEM, in order to build a block-circulant matrix, we consider an extended domain $\Omega^{\text{E}} = \prod\limits_{i = 1}^d [x_\textbf{0}[i], x_{\mm}[i]]$, where  $\mm = (\mm[1],\ldots,\mm[d])^T$ with $\mm[i] \geq 2 \NN[i]$ and $x_{\mm}[i] = x_\textbf{0}[i] + \mm[i]\hh[i]$ for $1\leq i \leq d$. Figure \ref{fig:2D_grid} shows an example of the computation domain $\Omega$ with $\NN = (4,4)^T$ and the extended domain $\Omega^E$ with $\mm = (8,8)^T$. Vectors $\jj_k \in \mathcal{I}(\mm+\mathbf{1})$ form the extended regular grid $\Omega_{\text{r}}^E = \{x_{\jj_k} = (x_{\jj_k}[1],\ldots,x_{\jj_k}[d])^T\mathrel{}|\mathrel{}\jj_k \in \mathcal{I}(\mm+\mathbf{1})\} \subset \Omega^{\text{E}}$. There are $\overline{\mm+\mathbf{1}}$ regular grid points in the set $\Omega_{\text{r}}^E$. The parallelepiped $\Omega^{\text{E}}$ can be divided into $d$-dimensional small parallelepipeds as (see also Fig.~\ref{fig:2D_grid}):
\begin{equation}
\Omega^{\text{E}} = \bigcup\limits_{\jj_k\in \mathcal{I}(\mm)} \Omega_{\jj_k},
\end{equation}
where $\Omega_{\jj_k}$, $\jj_k\in \mathcal{I}(\mm)$, are as in (\ref{eq:omjj}).

We now describe BCEM in the $d$-dimensional case, which is applicable to our block-Toeplitz covariance matrix $R$. In contrast to the 1D setting, where the covariance function is always even because of its symmetry, further classifications of covariance functions are needed in higher dimensional cases. We say that $r$ is \textit{component-wise even} in $i$th coordinate if
\begin{equation}
r(x[1],\ldots,-x[i],\ldots,x[d]) = r(x[1],\ldots,x[i]\ldots,x[d])
\label{even}
\end{equation}
for all $x \in \mathbf{R}^d$; otherwise, we say that $r$ is \textit{component-wise uneven} in some coordinates.

For simplicity of the exposition, let us assume for now that $r(x)$ is component-wise even in all coordinates (we will discuss a modification of BCEM in the uneven case in Remark~\ref{rk:uneven}).

We first build  the block circulant embedding of the block-Toeplitz matrix $R$. Consider the first row of the block circulant matrix $C$, which is an $\ell\times\ell\overline{\mm}$ matrix $C_f$ of the form
\begin{equation} \label{eq:first_block_circ}
C_f= \left[
\begin{array}{c c c c}
C_{\jj_0} & C_{\jj_1} & \cdots & C_{\jj_{\overline{m}-1}}
\end{array}
\right],
\end{equation}
where $(i,j)$-th element of $C_{\jj_k}$ is equal to
\begin{equation} \label{eq:cjb}
C_{\jj_k}^{i,j} = r(g_\mm(s^{(\jj_0)}_i - s^{(\jj_k)}_j))
\end{equation}
and the vector function $g_\mm=(g_\mm^1,\ldots,g_\mm^d)^T$ is defined by
\begin{equation}\label{eq:gmm}
g^i_\mm(x_\jj) = \begin{cases}
 x_{\jj}[i], & \text{if $|x_{\jj}[i]| < \mm[i]\hh[i]/2$,} \\
 \mm[i]\hh[i] - |x_{\jj}[i]|, & \text{otherwise.}
\end{cases}
\end{equation}
The block circulant matrix $C$ is generated by its first row $C_f$ in the usual way. Also note that
\begin{equation}
R = [C_{\jj_k}]_{\jj_k \in \mathcal{I}(\NN)}.
\end{equation}

The block circulant matrix $C$ is block diagonalizable by a block discrete Fourier transform (BDFT) matrix, $F_B = F_\mm^d \otimes I_\ell$, where $F_\mm^d$ is a $d$-dimensional DFT matrix. That is, we have $ C = (1/\overline{\mm}) F_B \Lambda F_B^H$, where $\Lambda = \mathrm{diag}(\Lambda_0,\ldots,\Lambda_\mm)$. The blocks on the diagonal of $\Lambda$ can be found by simply taking BDFT of first block row \cite{davis:1979, wood:1999}. Furthermore, using the fact that $F_B$ is the tensor product involving the identity matrix,
we derive the following component-wise computation:
\begin{equation}\label{eq:fft_pw}
[\Lambda_0^{i,j} \cdots \Lambda_{\overline{\mm}-1}^{i,j}] = \fftd([C_0^{i,j} \cdots C_{\overline{\mm}-1}^{i,j}]),
\end{equation}
where $1 \leq i,j \leq \ell$ and $\fftd$ is the $d$-dimensional FFT. Note that instead of $\mathrm{FFT}_1$ we will write FFT.

Due to the fact that $\Lambda$ is Hermitian and all diagonal entries of $\Lambda_k$ are the same, the required number of $\fftd$ of size $\mm$ (which is equivalent to FFT of order
$\bar \mm$) in (\ref{eq:fft_pw}) is $\ell(\ell+1)/2 - (\ell-1)$.

If $\Lambda$ is positive-definite, the Cholesky decomposition $\Lambda = LL^H$ exists, where $L$ is a block diagonal matrix with each block being a lower triangular matrix. Then, we obtain the decomposition $C = (1/\overline{\mm}) F_B L (F_B L)^H$.

As in the one-dimensional case (see Section~\ref{section:1d}), let $\xi = \xi_1 + i \xi_2$ be a complex-valued random vector of order $\overline{\mm}$ with $\xi_1$ and $\xi_2$ being real, normal random vectors such that $\EE[\xi_i] = 0$ and $\EE[\xi_i \xi_j^T] = \delta_{ij}I$. Set $\widetilde{L} = (1/\overline{\mm})^{1/2}L$ and $\eta = \widetilde{L}\xi$. Multiplying the square root of $C$ by $\xi$, we obtain the complex-valued vector
\begin{equation}
F_B \eta=\zeta = \zeta_1 + i \zeta_2,
\end{equation}
with the properties: $\EE[\zeta_1 \zeta_1^T] = \EE[\zeta_2 \zeta_2^T]=C$ and $\zeta_1$ and $\zeta_2$ are independent. Using tensor-product properties of $F_B$, $\zeta$ can be computed in the component-wise manner:
\begin{equation} \label{eq:componentwise}
[\zeta[i]\quad \zeta[i+\ell]\ldots \zeta[i+(\overline{\mm}-1)\ell] = \fftd([\eta[i]\quad \eta[i+\ell]\ldots \eta[i+(\overline{\mm}-1)\ell])
\end{equation}
for $1 \leq i \leq \ell$.

To summarize, the new BCEM can be presented in the algorithmic form as follows.

\smallskip

\begin{algorithm}[H]
	\caption{Block circulant embedding method (BCEM)}
	\label{alg:bcem}
Given $N\in \mathbb{Z}^d, x_\textbf{0}\in \Omega$, and strictly positive valued vector $\hh \in \mathbf{R}^d$,
	\begin{algorithmic}
	 	\STATE\textit{Step} 1. Choose a vector $\mm\in \mathbb{Z}^d$ such that $\mm[i] \geq 2 \NN[i]$ for all $1 \leq i \leq d$. \label{step01}
		\STATE\textit{Step} 2.  Compute the first block row of the circulant matrix $C$ as described in (\ref{eq:first_block_circ})-(\ref{eq:gmm}). \label{step02}
		\STATE\textit{Step} 3.  Compute the block diagonal matrix $\Lambda = \mathrm{diag}(\Lambda_0, \cdots,\Lambda_{\overline{\mm}-1})$ using (\ref{eq:fft_pw}). \label{step03}
		\STATE\textit{Step} 4.  Compute the square-root of $\Lambda$ applying Cholesky decompositions to diagonal blocks of $\Lambda$:
\begin{equation}
\Lambda = L L^H,
\end{equation}
where $L$ is a block diagonal matrix with lower triangular block of order $\ell$. \label{step04}
		\STATE\textit{Step} 5.   If the Cholesky decomposition fails in \textit{Step} 4, increase $\mm[i]$ by one or more and go to \textit{Step} 2. \label{step05}
		
		\STATE\textit{Step} 6.  Compute $\widetilde{L} = (1/\overline{\mm})^{1/2}L$. \label{step06}
		\STATE\textit{Step} 7. Generate 	a random complex vector of dimension $\overline{\mm}\ell, \xi = \xi_1 + i \xi_2$, with two independent vectors $\xi_1$ and $\xi_2$ being $\mathcal{N}(0,I_{\overline{\mm}\ell})$. Compute $\eta = \widetilde{L}\xi$. \label{step07}
		\STATE\textit{Step} 8. Compute $z = (\zeta[1],\ldots,\zeta[\overline{\mm}\ell])^T$ by applying ç times as in (\ref{eq:componentwise}). \label{step08}
	\end{algorithmic}
\end{algorithm}

\smallskip

Note that if $\ell = 1$, then $\Omega_r$ is regular, $C$ is circular, $\Lambda$ becomes diagonal instead of block diagonal and Algorithm~\ref{alg:bcem} degenerates to the standard CEM.

\begin{remark}\label{rk:uneven}
Algorithm~\ref{alg:bcem} is applicable when the covariance function is component-wise even (\update{2.2}{see (\ref{even})}). Although the covariance function is even by definition, i.e., $r(-x) = r(x)$, it could be component-wise uneven, e.g.,
\[
r(x) = \exp(-x^TAx)  \quad \text{ with } \quad A =
\left[
\begin{array}{c c}
2 & -1 \\
-1 & 2
\end{array}
\right], \,\,
x \in \R^2,
 \]
 is uneven.
 In this case, the matrix $C$ defined by the vector function $g_m$ in (\ref{eq:gmm}) usually does not have block circulant structure because of conflicting definitions at the points $x$ with $x[i] = \mm[i]\hh[i]/2$, if $r$ is uneven in $i$th coordinate. Two adjustments to make CEM work in uneven cases were suggested in \cite{wood:1994}, which can be applied to BCEM by modifying Algorithm~\ref{alg:bcem} as follows.
\\
If $r$ is uneven in the $i$ th coordinate, either
\begin{itemize}
\item[(a)]  choose $\mm[i]$ to be an odd integer, e.g., a power of three;
\end{itemize}
or
\begin{itemize}
\item[(b)]  still choose $\mm[i]$ to be an even integer and define $C_{\jj_k}$ using (\ref{eq:cjb}) and  (\ref{eq:gmm}), except put $r(x) = 0$ for all $x$ such that $|x[i]| = \mm[i]\hh[i]/2$ for some $i$.
\end{itemize}
In either case, the resulting matrix $C$ has a block circulant structure and thus  Algorithm \ref{alg:bcem} can be seamlessly extended to the uneven case with aforementioned modifications.
\end{remark}

\begin{remark}
As mentioned earlier, the matrix $C$ is often negative definite in practical applications of CEM and BCEM. Following \cite{wood:1994},  we increase the matrix $C$ in Algorithm~\ref{alg:bcem} (see its Step 5) until it becomes non-negative definite (the padding technique). The padding technique is universal and usually efficient when the correlation length of a random field is in a range from small to medium relative to the size of a computational domain and the field is not too smooth. Otherwise, the use of the padding technique could be very expensive. There are two recently developed alternatives to padding (a cut-off of the circulant matrix \cite{stein:2002,gneiting:2006} and smoothing window circulant embedding \cite{helgason:2014}), which can deal with the problem of negative definiteness of circulant matrices effectively. The techniques from \cite{stein:2002,gneiting:2006,helgason:2014} are applicable to BCEM as they are for CEM.
\label{rempad}
\end{remark}

The equispaced FFT is highly parallelizable in high dimensions, and its highly scalable implementations are proposed in \cite{frigo:1998, eleftheriou:2005}. This could be beneficial in the standard CEM because its computation is dominated by the FFT.  Still equipped with the parallelism of the FFT, BCEM can be further parallelized in a natural way because the applications of $\fftd$ in \textit{Step} \ref{step03} and \textit{Step} \ref{step08} of Algorithm~\ref{alg:bcem} can be performed separately and simultaneously. Moreover, block-diagonal matrix operations in \textit{Step} \ref{step04} and \textit{Step} \ref{step07} can be performed separately and simultaneously. Therefore, the overall BCEM algorithm contains two-level parallelism, giving us significant advantage over the standard CEM.

As we will see in the next section, BCEM can be faster than CEM both in taking square-roots of the corresponding circulant matrices (performed, of course, only once per the whole Monte Carlo simulation) and in sampling the random field required in each Monte Carlo run. The latter is usually more important in Monte Carlo-type simulations.

\section{Computational Complexity of BCEM}\label{sec:comp}
In this section we analyze the computational complexity of BCEM. To this end, we use the same convention as in Golub and Van Loan \cite{golub:1996} for counting the number of floating point operations: $5m\log_2m$ flops for FFT of size $m$ and $n^3/3$ flops for the Cholesky decomposition of a matrix of order $n$.

\textit{Step} \ref{step03} of Algorithm~\ref{alg:bcem} is the initial factorization of the block circulant matrix $C$ by taking BDFT of its first block row which can be computed using the ordinary DFT in (\ref{eq:fft_pw}) at the cost
\begin{equation}\label{eq:cost1}
 \text{cost}_1 = \left(\frac{\ell(\ell+1)}{2} - (\ell-1)\right) (5 \overline{\mm}\log_2\overline{\mm}) \text{ flops}.
\end{equation}
Here we took into account that each $\Lambda_k$ is Hermitian and its diagonal elements have the same value.

In \textit{Step} \ref{step04}, the square-root operation on the block diagonal matrix $\Lambda$ with $\overline{\mm}$ blocks of order $\ell$ can be performed on each block separately using the Cholesky decomposition method. In Proposition~\ref{prop:symmetricity}, we proved in the one-dimensional case
that $\Lambda$ has pairs of complex-conjugate blocks, $\Lambda_k$ and $\Lambda_{m-k}$, which allows us to compute the square-root of $\Lambda_k$ and use its complex-conjugate as a square-root of its complex-conjugate pair $\Lambda_{m-k}$. This is based on the periodicity and conjugate symmetry of FFT. Hence, Proposition~\ref{prop:symmetricity} can be extended to the higher dimensional cases. Then the matrix $\Lambda$ can be decomposed at the cost
\begin{equation}\label{eq:cost2}
\text{cost}_2 = \prod\limits_{i = 1}^d \left( \frac{\mm[i]}{2}+1\right)\frac{\ell^3}{3} \text{ flops}.
\end{equation}

\smallskip

\begin{remark}
Note that if the nodes of $S_{\jj_k}$ are regularly (uniformly) distributed in $\Omega_{\jj_k}$ for all $\jj_k \in I(\mm)$ which is often the case in applications (see, e.g., Example~1.1), then all blocks on the diagonal of $\Lambda$ are block Toeplitz (Toeplitz). Toeplitz matrix  and block Toeplitz matrix can be decomposed using Schur's algorithm \cite{stewart:1997} and block Schur's algorithm \cite{gallivan:1996}, respectively, which have $O(\ell^2)$ complexity as opposed to $O(\ell^3)$
for the standard Cholesky decomposition. Making use of Schur's algorithms can reduce the cost of Algorithm~\ref{alg:bcem}.
\end{remark}

\smallskip

In \textit{Step} \ref{step08}, computing a realization of $\zeta$ requires block diagonal matrix-vector multiplication $\widetilde{L}\xi$ and $\ell$ applications of FFT of order $\overline{\mm}$ in (\ref{eq:componentwise}) at the cost
\begin{equation}\label{eq:cost3}
\text{cost}_3 = \ell^2\overline{\mm} \text{ flops}
\end{equation}
and
\begin{equation}\label{eq:cost4}
\text{cost}_4 = \ell (5 \overline{\mm} \log_2 \overline{\mm}) \text{ flops},
\end{equation}
respectively.

To conclude, the cost of BCEM is $\mathcal{O}(\ell^3\overline{\mm}+\ell^2 \overline{\mm} \log_2 \overline{\mm})$ flops. In practical applications of BCEM (see, e.g. examples in the Introduction) the size of blocks $\ell$ is relatively small while the number of blocks $\overline{\mm}$ is large. Recall that BCEM is designed for block-regular grids $\Omega_s$. Its main computational advantage in comparison with CEM (which is designed for regular grids) comes from the fact that the use of CEM in the case of simulations on a block-regular grid $\Omega_s$ requires regularization of $\Omega_s$, i.e., adding a significant number of extra nodes which BCEM does not need. Hence BCEM works on a grid with a smaller number of nodes than CEM and needs to generate random vectors $\zeta$ of smaller size than CEM (and hence makes less number of calls to a random number generator to sample $\xi$).

\begin{remark}
It can be shown that the use of BCEM on a {\it regular} grid split in blocks of a size $\ell$ can be more effective in sampling the random field than CEM
but it is computationally more expensive in the matrix decomposition than CEM. The latter can be overcome by exploiting the fact that BCEM is parallelizable in comparison with CEM. Thus BCEM can be more effective than CEM even in the case of regular grids for which CEM is designed. \label{rem41}
\end{remark}

We now compare  computational complexity of BCEM and CEM using the first two examples from the Introduction and the following exponential covariance function (cf. (\ref{eq:cjb})):
\begin{equation}\label{eq:exp_cov}
r(\textbf{x}) = \sigma^2
\mathrm{exp}\left(-\frac{\|\mathbf{x}\|_1}{0.3}\right),
\end{equation}
where $\|\cdot \|_1$ means $\mathrm{L}_1$-norm.
 We note that the circulant matrix $C$ (cf. (\ref{eq:first_block_circ}), (\ref{eq:cjb})) of the size $\mm = 2\NN$ formed by (\ref{eq:exp_cov}) is always positive definite
(see, e.g. \cite{dietrich:1997}).  This means, in particular, that Step \ref{step05} (i.e., padding) of Algorithm~\ref{alg:bcem} is not needed in this case. For simplicity, we consider the domain $\Omega$ to be the unit square in the examples.

\smallskip

\begin{figure}[ht!]
\centering
\includegraphics[width=\textwidth]{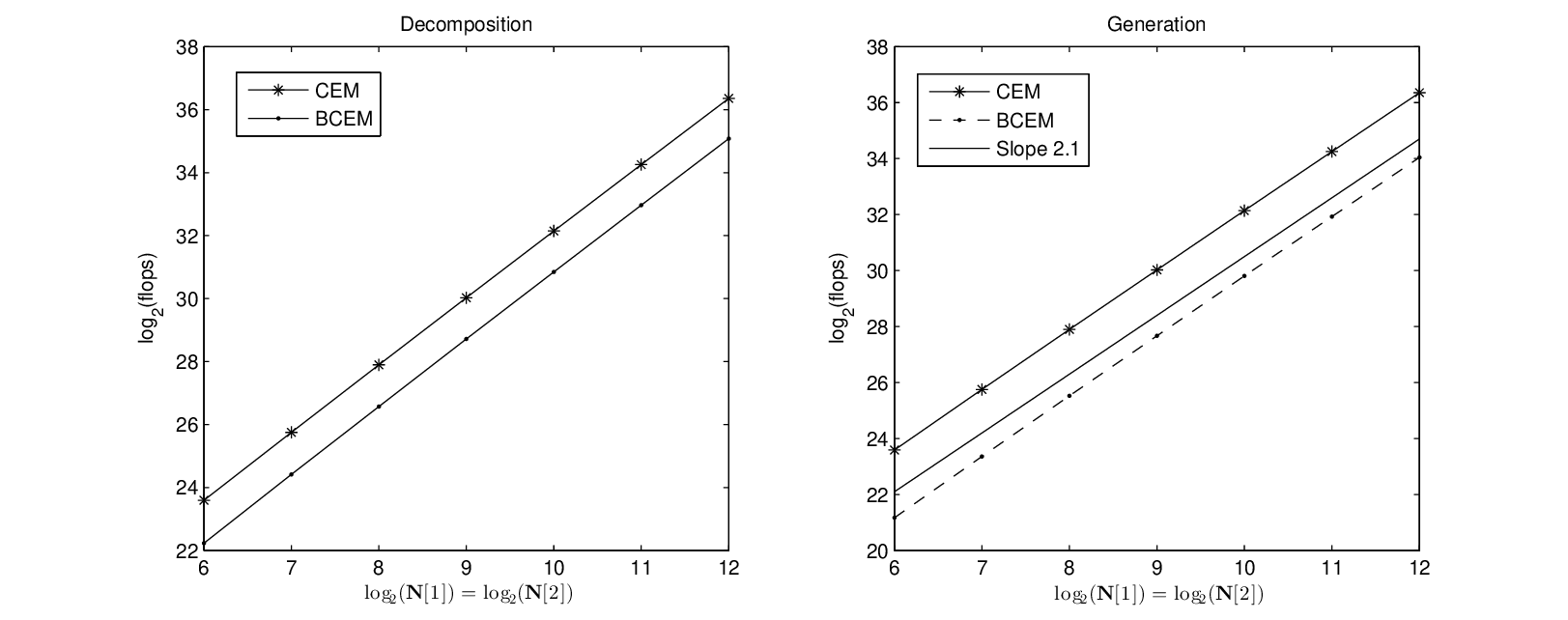}
\caption{The floating point operations required for the matrix decomposition and random field generation stages of CEM and BCEM in Example~\ref{ex:agfem2d}.}
\label{fig:fem_costs}
\end{figure}

\begin{example}\label{ex:agfem2d}
{\it Triangular finite element with a quadrature point located at the barycentre of the triangle.}
\end{example}
\vspace{-8pt}
\noindent In Example~\ref{ex:fem2d} (see Figure \ref{fig:2D_FEM}), each rectangular block contains 9 uniform grid nodes. Hence the order of the circulant matrix used by CEM is $9\overline{\mm}$, where $\mm = 2\NN$ and  $\overline{\mm}$ is the number of rectangular blocks in the extended domain $\Omega^\mathrm{E}$. Then the matrix decomposition cost for CEM is $45\overline{\mm}\log_2 9\overline{\mm}$ flops, and generation of each realization of the random field requires another $45\overline{\mm}\log_29\overline{\mm}$  flops.

Here BCEM uses only two points in each rectangular block, so the order of the block-circulant matrix is $2\overline{\mm}$. Substituting $\ell = 2$ into (\ref{eq:cost1}) and (\ref{eq:cost2}), the total matrix decomposition cost for BCEM is $10\overline{\mm}\log_2\overline{\mm}+(\mm[1]/2+1)(\mm[2]/2+1)(8/3)$ flops. Each realization of the random field is generated at the cost of $4\overline{\mm} +   10\overline{\mm}\log_2\overline{\mm}$ flops (see (\ref{eq:cost3}) and (\ref{eq:cost4})).

Figure \ref{fig:fem_costs} shows the floating point operations required by the two algorithms. One can see that BCEM is more effective in both procedures and that the complexity of BCEM grows at roughly the same rate as for CEM. Compared to CEM, BCEM reduces the matrix decomposition cost and the generation cost approximately in $2.5$ time and $4$ time, respectively. The improvement in computational efficiency is due to the fact that BCEM works with just 2/9 of nodes that CEM uses to build the circulant matrix. This also means that BCEM requires 4.5 time less memory than CEM.

\begin{figure}[!ht]
    \centering
    \includegraphics[width=.45\textwidth]{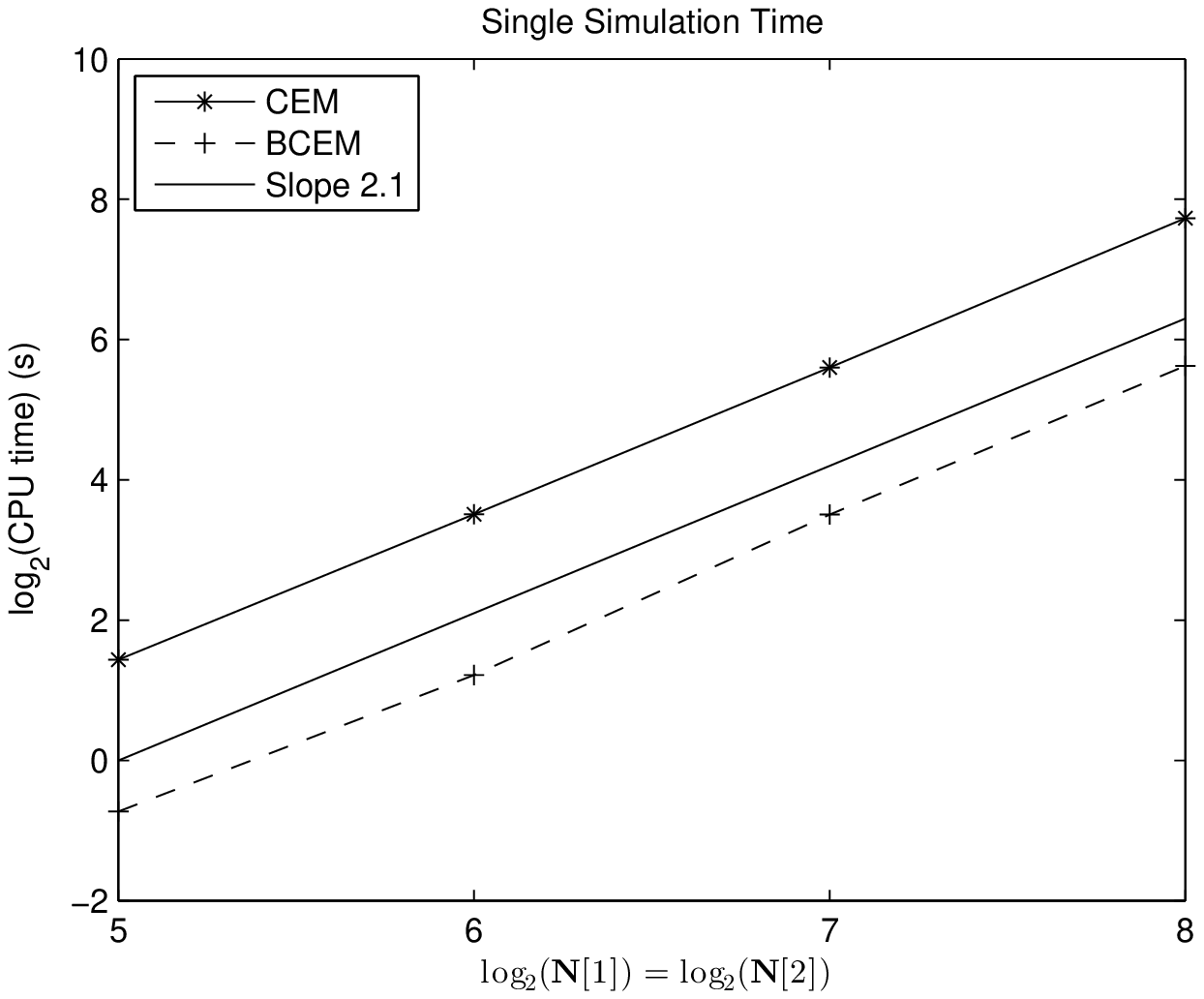}
    \qquad
    \begin{tabular}[b]{c|rr||c}\hline
     & \multicolumn{3}{c}{CPU Time (s)}  \\ \cline{2-4}
      $N$ & BCEM & CEM & speed-up \\ \hline
      32& 0.60&2.71  & 4.5\\
      64& 2.32&11.39 & 4.9 \\
      128& 11.34&48.55 & 4.3\\
      256&  49.37&212.21 &4.3	\\\hline
      \multicolumn{4}{c}{}\\
      \multicolumn{4}{c}{}\\
      \multicolumn{4}{c}{}\\
      \multicolumn{4}{c}{}\\
      \multicolumn{4}{c}{}\\
    \end{tabular}
    \captionlistentry[table]{Average time required to simulate a single realization of the random field in Example \ref{ex:agfem2d} computed using 1000 independent samples.}\label{tab:fem_generation}
    \captionsetup{labelformat=andtable}
    \caption{Average time required to simulate a single realization of the random field in Example \ref{ex:agfem2d} computed using 1000 independent samples.}\label{fig:fem_generation}
  \end{figure}

To compare performance of BCEM and CEM further, we generated samples of the random field by these two methods on an Intel Xeon E5-2450, 96GB RAM computer using MATLAB R2014a.
Figure \ref{fig:fem_generation} shows the average computational time of generation of a single realization of the random field by both methods and how it increases with increasing $\NN$. Table~\ref{tab:fem_generation} gives the CPU time and the speed up in generating a random vector using BCEM against the ones using CEM. For both methods, the CPU time increases with increase of $\NN$ at about the same rate as the theoretical rate shown in Figure \ref{fig:fem_costs} (right). We see that BCEM is about $4.3-4.5$ faster than CEM, which is close to the theoretical cost estimation in Figure \ref{fig:fem_costs}.

\begin{figure}[ht!]
\centering
\includegraphics[width=1\textwidth]{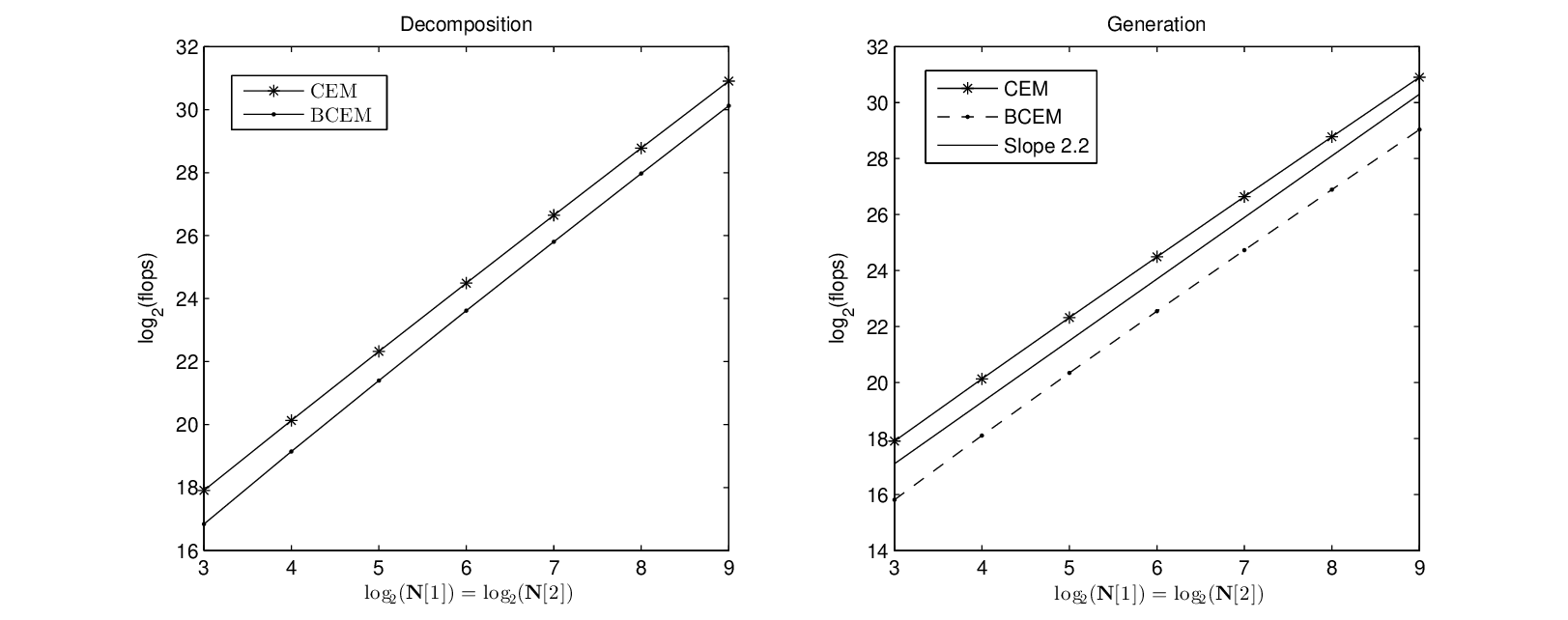}
\caption{The floating point operations required for the matrix decomposition and random field generation stages of CEM and BCEM in Example~\ref{ex:agmlmc}.}
\label{fig:mlmc_costs}
\end{figure}

\smallskip

\begin{example}\label{ex:agmlmc}
{\it Cell-centered finite volume discretization in multilevel Monte Carlo (MLMC) computation.}
\end{example}
\vspace{-8pt}
\noindent In Example~\ref{ex:mlmc}, BCEM uses 5 out of 16 uniform nodes required for CEM in each individual block to generate random variables located at the centers of both the fine and coarse cells. That is, CEM should generate random variables at the extra 11 nodes that are not used in the finite volume discretization and are not used by BCEM. Then memory requirement for CEM and BCEM are $16\overline{\mm}$ and $5\overline{\mm}$, respectively, which makes BCEM more attractive when the number of blocks is large.

Whereas the matrix decomposition and sampling costs in CEM both require $80\overline{\mm}\log_216\overline{\mm}$ flops, the computational costs of the matrix decomposition and sampling in BCEM are $55\overline{\mm}\log_2\overline{\mm}+(\mm[1]/2+1)(\mm[2]/2+1)(125/3)$ flops and $25\overline{\mm} +   25\overline{\mm}\log_2\overline{\mm}$ flops, respectively (see (\ref{eq:cost1}) - (\ref{eq:cost4}) with $\ell = 5$). Note that the total costs are dominated by $\overline{\mm}\log_2\overline{\mm}$. Hence, for large $\overline{\mm}$, the ratio of the matrix decomposition in CEM to one in BCEM is close to $80/55\approx1.45$ . For the sample generation cost, the ratio is close to $80/25\approx3.2$. These theoretical computational costs are shown in Figure~\ref{fig:mlmc_costs}.

\begin{figure}[!ht]
    \centering
    \includegraphics[width=.45\textwidth]{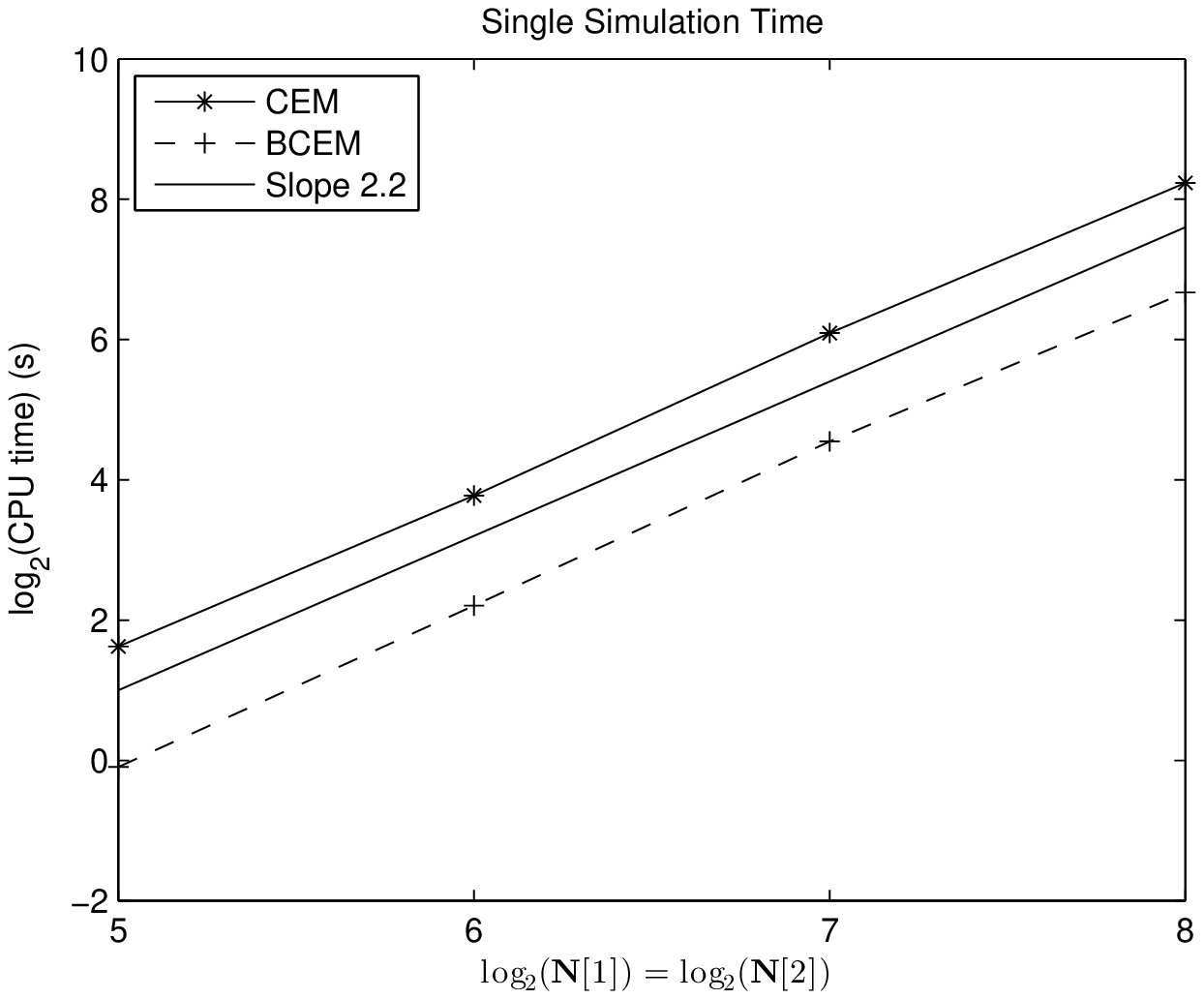}
    \qquad
    \begin{tabular}[b]{c|rr||c}\hline
     & \multicolumn{3}{c}{CPU Time (s)} \\ \cline{2-4}
      $N$ & CEM & BCEM & speed-up\\ \hline
      32& 0.94& 3.08&3.3\\
      64& 4.62& 13.71&3.0\\
      128&  23.37& 68.28&2.9 \\
      256& 101.82& 301.24 & 3.0\\\hline
      \multicolumn{3}{c}{}\\
      \multicolumn{3}{c}{}\\
      \multicolumn{3}{c}{}\\
      \multicolumn{3}{c}{}\\
      \multicolumn{3}{c}{}\\
    \end{tabular}
    \captionlistentry[table]{Average time required to simulate a single realization of the random field in Example \ref{ex:agmlmc} computed using 1000 independent samples.}\label{tab:mlmc_generation}
    \captionsetup{labelformat=andtable}
    \caption{Average time required to simulate a single realization of the random field in Example \ref{ex:agmlmc} computed using 1000 independent samples.}\label{fig:mlmc_generation}
  \end{figure}

 Figure \ref{fig:mlmc_generation} {gives the CPU times for the random field generation. We see that} the actual computational cost increases with increase of $\NN$ similarly to the theoretical one as in Figure~\ref{fig:mlmc_costs}. Table \ref{tab:mlmc_generation} demonstrates that BCEM is nearly 3 time faster than CEM as we expected from Table \ref{tab:mlmc} and Figure~ \ref{fig:mlmc_costs}. Also note that BCEM is highly parallelizable, so the computation cost can be further reduced using parallel algorithms.

We have compared BCEM and CEM on the 2D examples here. It is not difficult to see (cf. Table~\ref{tab:mlmc}) that in 3D cases BCEM can outperform CEM even more dramatically.

\begin{remark}
The MATLAB codes for BCEM used for Examples~\ref{ex:agfem2d} and~\ref{ex:agmlmc} are available at \newline
\textit{https://www.maths.nottingham.ac.uk/personal/pmzmp/bcem.html}.
\end{remark}

\section*{Acknowledegments}
This work was partially supported by the EPSRC grant EP/K031430/1.

\end{document}